\documentclass{article}

\usepackage[utf8]{inputenc}

\usepackage{fixltx2e}

\usepackage{amsmath,amssymb,amsthm}

\usepackage{cite}

\usepackage{nameref,hyperref}

\usepackage{authblk}

\usepackage{placeins}

\usepackage{graphicx}
\usepackage{caption}
\usepackage{subcaption}

\newtheorem{theorem}{Theorem}
\newtheorem{lemma}[theorem]{Lemma}

\newcommand{\low}[1]{\tilde{{#1}}}
\newcommand{\up}[1]{\bar{{#1}}}
\newcommand{\mysum}[2]{\Sigma^{{#1}}_{#2}}

\usepackage[ruled]{algorithm2e}

\begin{document}
	
\title{An Optimal Realization Algorithm for Bipartite Graphs with Degrees in Prescribed Intervals\footnote{Submitted to the Journal of Discrete Algorithms.}}
\author{Steffen Rechner}
\affil{Institute of Computer Science\\Martin Luther University Halle-Wittenberg, Germany}
\affil{\small{\texttt{steffen.rechner@informatik.uni-halle.de}}}
\date{\today}

\maketitle

\begin{abstract}
We consider the problem of constructing a bipartite graph whose degrees lie in prescribed intervals. Necessary and sufficient conditions for the existence of such graphs are well-known. However, existing realization algorithms suffer from large running times. In this paper, we present a realization algorithm that constructs an appropriate bipartite graph~$G=(U,V,E)$ in~$\mathcal{O}(|U| + |V| + |E|)$ time, which is asymptotically optimal.
In addition, we show that our algorithm produces edge-minimal bipartite graphs and that it can easily be modified to construct edge-maximal graphs.
\end{abstract}

\section{Introduction}

The construction of bipartite graphs with prescribed degrees is a well-studied algorithmic problem with various applications in science.
Whereas the classical problem asks for a bipartite graph whose degrees must \emph{exactly} match prescribed integers, we study the problem of constructing a bipartite graph whose degrees have to lie in prescribed \emph{intervals}.
As one application, the construction of such graphs is part of several \emph{sampling algorithms} that aim at producing random bipartite graphs whose degrees lie in the prescribed intervals. Such sampling algorithms typically work by randomly modifying an existing bipartite graph, while preserving the prescribed bounds on the degrees~\cite{rechner2017interval}.
If no bipartite graph with prescribed degrees is given, the first step of such algorithms is to construct an initial bipartite graph.
In such cases, the efficient realization is of high importance for the efficiency of the whole sampling algorithm.


\paragraph{Problem Definition}

Given a four-tuple of integer vectors~$(\low{\mathbf{r}}, \up{\mathbf{r}}, \low{\mathbf{c}}, \up{\mathbf{c}})$ with~$\low{\mathbf{r}} = (\low{r}_1, \low{r}_2, \ldots, \low{r}_{n_1})$,~$\up{\mathbf{r}} = (\up{r}_1, \up{r}_2, \ldots, \up{r}_{n_1})$,~$\low{\mathbf{c}}=(\low{c}_1, \low{c}_2, \ldots, \low{c}_{n_2})$,~$\up{\mathbf{c}} = (\up{c}_1, \up{c}_2, \ldots, \up{c}_{n_2})$, we want to construct a bipartite graph whose vertex degrees are bounded from below by the vectors~$\low{\mathbf{r}}$ and~$\low{\mathbf{c}}$ and from above by~$\up{\mathbf{r}}$ and~$\up{\mathbf{c}}$.
More precisely, let~$d_G \colon U \cup V \to \mathbb{N}$ describe the vertex degrees of a bipartite graph~$G=(U,V,E)$ with disjoint vertex sets~$U=\{ u_1, u_2, \ldots, u_{n_1} \}~$ and~$V = \{ v_1, v_2, \ldots, v_{n_2} \}$. 
We seek a bipartite graph
such that~$\low{r}_i \leq d_G(u_i) \leq \up{r}_i$ and~$\low{c}_j \leq d_G(v_j) \leq \up{c}_j$ hold for all~$i$ in range~$1 \leq i \leq n_1$ and~$j$ in range~$1 \leq j \leq n_2$. 

\paragraph{Related Work}

In the special case of~$\low{\mathbf{r}} = \up{\mathbf{r}}$ and~$\low{\mathbf{c}} = \up{\mathbf{c}}$, Ryser's algorithm~\cite{ryser1957combinatorial} is used for decades to construct a valid bipartite graph~$G=(U,V,E)$ in~$\mathcal{O}(|U| + |V| + |E|)$ time.
In the general case, the problem can be solved by finding a~$(g,f)$-factor in the complete bipartite graph with~$g$ and~$f$ being integer vectors obtained by concatenating the vectors of lower and upper bounds.
As~$(g,f)$-factors in arbitrary graphs~$\mathcal{G} = (\mathcal{V}, \mathcal{E})$ can be constructed in~$\mathcal{O}(|\mathcal{V}|^3)$ time~\cite{anstee1985algorithmic}, this approach leads to a realization algorithm with a running time of~$\mathcal{O}((|U|+|V|)^3)$.
Another way of finding a valid bipartite graph is to construct an appropriate flow network~$\mathcal{F} = (\mathcal{V}, \mathcal{A})$ and to calculate a maximal flow~\cite{fulkerson1959network}. 
As calculating a maximal flow can be achieved in~$\mathcal{O}(|\mathcal{V}|\cdot |\mathcal{A}|)$ time~\cite{orlin2013max} and, in our case, the flow network needs to have~$|\mathcal{V}| = |U|+|V|+2$ vertices and~$|\mathcal{A}| =  |U| \cdot |V| + |U| + |V|$ arcs, we gain a running time of~$\mathcal{O}(|U| \cdot |V| \cdot (|U|+|V|))$.

\paragraph{Contribution}
In this article, we present a realization algorithm whose running time is bounded by~$\mathcal{O}(|U| + |V| + |E|)$ with~$|E|$ being the number of edges of the bipartite graph that is been constructed. 
After summarizing the necessary theoretical background, we present the realization algorithm and show that its running time is~$\mathcal{O}(|U| + |V| + |E|)$. Finally, we show that the bipartite graph constructed by our method is edge-minimal.

\section{Preliminaries}

\paragraph{Definitions}
	Let~$n_1$ and~$n_2$ be positive integers. Let~$\mathbf{r} = (r_1, r_2, \ldots, r_{n_1})$,~$\bar{\mathbf{r}} = (\bar{r}_1, \bar{r}_2, \ldots, \bar{r}_{n_1})$, and~$\mathbf{c} = (c_1, c_2, \ldots, c_{n_2})$ be integer vectors of length~$n_1$ respectively~$n_2$.
	The integer vector~$\mathbf{r}$ is called \emph{non-increasing} if~$r_i \geq r_{i+1}$ for~$1 \leq i < n_1$.
	We say that the pair of integer vectors~$(\mathbf{r}, \mathbf{c})$ is \emph{bi-graphical} if and only if there is a bipartite graph~$G=(U,V,E)$ such that~$d_G(u_i) = r_i$ and~$d_G(v_j) = c_j$ for~$1\leq i \leq n_1$ and~$1 \leq j \leq n_2$. 
	We will write~$\mathbf{r} \leq \bar{\mathbf{r}}$ if and only if~$r_i \leq \bar{r}_i$ for all~$i$ in range~$1 \leq i \leq n_1$.
	The integer vector~$\mathbf{c}' = (c_1', c_2', \ldots, c_{n_1}')$ of length~$n_1$ is called \emph{conjugate} vector of~$\mathbf{c}$ if and only if~$c_i' = |\{ j \colon 1 \leq j \leq n_2 \wedge c_j \geq i  \}|$ for each~$i$ in range~$1\leq i \leq n_1$.
	We abbreviate the sum~$\sum_{i=1}^k r_i$ by~$\mysum{k}{\mathbf{r}}$. We say that the conjugate vector~$\mathbf{c}'$ \emph{dominates}~$\mathbf{r}$ and write~$\mathbf{r} \trianglelefteq \mathbf{c}'$ if and only if~$\mysum{k}{\mathbf{r}} \leq \mysum{k}{\mathbf{c}'}$ holds for~$1 \leq k \leq n_1$.
	In addition, we define~$\mysum{}{\mathbf{r}}$ to be the vector of \emph{partial sums} of an integer vector~$\mathbf{r}$, thus~$\mysum{}{\mathbf{r}} = (\mysum{1}{\mathbf{r}}, \mysum{2}{\mathbf{r}}, \ldots, \mysum{n_1}{\mathbf{r}})$.
	Finally, we say that the four-tuple of integer vectors~$(\low{\mathbf{r}}, \up{\mathbf{r}}, \low{\mathbf{c}}, \up{\mathbf{c}})$ is \emph{realizable} if and only if there is a bipartite graph~$G=(U,V,E)$ such that~$\low{r}_i \leq d_G(u_i) \leq \up{r}_i$ and~$\low{c}_j \leq d_G(v_j) \leq \up{c}_j$ hold for~$1 \leq i \leq n_1$ and~$1 \leq j \leq n_2$. 
	In such cases, we will call~$G$ a \emph{realization} of the four-tuple~$(\low{\mathbf{r}}, \up{\mathbf{r}}, \low{\mathbf{c}}, \up{\mathbf{c}})$.\\

Our algorithm is based on the following well-known theorems.

\begin{theorem}[Gale~\cite{gale1957theorem}, Ryser\cite{ ryser1957combinatorial}]
	Let~$\mathbf{r} = (r_1, r_2, \ldots, r_{n_1})$ be a non\--in\-creas\-ing integer vector and let~$\mathbf{c} = (c_1, c_2, \ldots, c_{n_2})$ be an integer vector.
	The pair~$(\mathbf{r}, \mathbf{c})$ is bi-graphical if and only if~$\mysum{n_1}{\mathbf{r}} = \mysum{n_2}{\mathbf{c}}$ and~$\mathbf{r} \trianglelefteq \mathbf{c}'$.
	\label{Theo:GaleRyser}
\end{theorem}

\begin{theorem}[Fulkerson~\cite{fulkerson1959network}, Schocker~\cite{schocker2001graphs}]
	Let~$\low{\mathbf{r}} = (\low{r}_1, \low{r}_2, \ldots, \low{r}_{n_1})$ and~$\low{\mathbf{c}} = (\low{c}_1, \low{c}_1, \ldots, \low{c}_{n_2})$ be non-increasing integer vectors, and let~$\up{\mathbf{r}} = (\up{r}_1, \up{r}_2, \ldots, \up{r}_{n_1})$ and~$\up{\mathbf{c}} = (\up{c}_1, \up{c}_2, \ldots, \up{c}_{n_2})$ be integer vectors with~$\low{\mathbf{r}} \leq \up{\mathbf{r}}$ and~$\low{\mathbf{c}} \leq \up{\mathbf{c}}$.
	The four-
	tuple~$(\low{\mathbf{r}}, \up{\mathbf{r}}, \low{\mathbf{c}}, \up{\mathbf{c}})$ is realizable if and only if~$\low{\mathbf{r}} \trianglelefteq \up{\mathbf{c}}'$ and~$\low{\mathbf{c}} \trianglelefteq \up{\mathbf{r}}'$.
	\label{Theo:Schocker}
\end{theorem}


\section{Realization Algorithm}

Our algorithm assumes that the integer vectors~$\low{\mathbf{r}}$ and~$\low{\mathbf{c}}$ are ordered non\--in\-creas\-ingly, which can be easily arranged by descendingly sorting the pairs~$(\low{\mathbf{r}} ,\up{\mathbf{r}})$ and~$(\low{\mathbf{c}} ,\up{\mathbf{c}})$ by their associated lower bounds.
In addition, the algorithm assumes that the four-tuple~$(\low{\mathbf{r}}, \up{\mathbf{r}}, \low{\mathbf{c}}, \up{\mathbf{c}})$ is realizable, which can be verified by Theorem~\ref{Theo:Schocker}. 
The key idea is now to iteratively construct a bi-graphical pair of integer vectors~$(\mathbf{r}, \mathbf{c})$ bounded by~$\low{\mathbf{r}} \leq \mathbf{r} \leq \up{\mathbf{r}}$ and~$\low{\mathbf{c}} \leq \mathbf{c} \leq \up{\mathbf{c}}$ that is afterwards realized via Ryser's algorithm~\cite{gale1957theorem}. The algorithm is divided into two parts. 


\subsection{Phase One}

Since the four-tuple~$(\low{\mathbf{r}}, \up{\mathbf{r}}, \low{\mathbf{c}}, \up{\mathbf{c}})$ is realizable by assumption, there is a bipartite graph~$G=(U,V,E)$ such that~$\low{c}_j \leq d_G(v_j) \leq \up{c}_j$ holds for each~$v_j \in V$. As a consequence, there must be an
integer vector~$\mathbf{c} = ( c_1, c_2, \ldots, c_{n_2})$ which describes the degrees of vertex set~$V$.
With other words, there is an integer vector~$\mathbf{c}$ bounded by~$\low{\mathbf{c}} \leq \mathbf{c} \leq \up{\mathbf{c}}$ such that~$(\low{\mathbf{r}}, \up{\mathbf{r}}, \mathbf{c}, \mathbf{c})$ is realizable.
In its first phase~(see Alg.~\ref{Alg:PhaseOne}), the algorithm constructs such a vector~$\mathbf{c}$.
For this purpose,~$\mathbf{c}$ is initialized with~$\low{\mathbf{c}}$. In a series of iterations, the algorithm identifies the right-most component~$c_i$ with~$c_i < \up{c}_i$ and increments the left-most components~$c_j$ with~$c_j = c_i$.
After a well-chosen number~$\delta_1$ of iterations, the algorithm returns the realizable four-tuple~$(\low{\mathbf{r}}, \up{\mathbf{r}}, \mathbf{c}, \mathbf{c})$.

\begin{algorithm}
	\LinesNumbered
	\DontPrintSemicolon
	\KwIn{realizable four-tuple~$(\low{\mathbf{r}}, \up{\mathbf{r}}, \low{\mathbf{c}}, \up{\mathbf{c}})$}
	\KwOut{realizable four-tuple~$(\low{\mathbf{r}}, \up{\mathbf{r}}, \mathbf{c}, \mathbf{c})$ with~$\low{\mathbf{c}} \leq \mathbf{c} \leq \up{\mathbf{c}}$}
	\caption{\textsc{Phase One}\label{Alg:PhaseOne}}
	$\mathbf{c} \gets \low{\mathbf{c}}$\tcp*{initialize~$\mathbf{c}$ with lower bounds~$\low{\mathbf{c}}$}
	$\delta_1 \gets \max\{ \mysum{j}{\low{\mathbf{r}}} - \mysum{j}{\low{\mathbf{c}}'} \colon 1 \leq j \leq n_1 \}$\tcp*{calculate number of steps}
	$i \gets n_2$\tcp*{right-most position such that~$c_i < \up{c}_i$}
	\For{$k=1,2,\ldots, \delta_1$}{
		\While(\tcp*[f]{proceed to next position with~$c_i < \up{c}_i$}){$c_i = \up{c}_i$}{
			$i \gets i - 1$\;
		}
		$j \gets \min \{ \ell \colon c_{\ell} = c_i\}$ \tcp*{identify left-most~$c_j$ with~$c_j = c_i$}
		swap~$\up{c}_i$ and~$\up{c}_j$\;
		$c_j \gets c_j + 1$\;
	}
	\Return{$(\low{\mathbf{r}}, \up{\mathbf{r}}, \mathbf{c}, \mathbf{c})$}\;
\end{algorithm}	

\paragraph{Example}

Consider the following integer vectors:
\begin{alignat*}{2}
\low{\mathbf{r}} &= (4,1,0) \quad &\low{\mathbf{c}} &= (2,2,0,0,0)\\
\up{\mathbf{r}}  &= (4,2,3) \quad &\up{\mathbf{c}}  &= (2,3,1,2,2).
\end{alignat*}
By setting up at the corresponding conjugate vectors, we observe that the four-tuple~$(\low{\mathbf{r}}, \up{\mathbf{r}}, \low{\mathbf{c}}, \up{\mathbf{c}})$ is realizable via Theorem~\ref{Theo:Schocker}. 
\begin{alignat*}{5}
\low{\mathbf{r}}  &= (4,1,0) \quad
&\up{\mathbf{r}}' &= (3,3,2,1,0) \quad
&\low{\mathbf{c}} &= (2,2,0,0,0) \quad
&\up{\mathbf{c}}' &= (5,4,1)
\\
\mysum{}{\low{\mathbf{r}}} &= (4,5,5) \quad
&\mysum{}{\up{\mathbf{r}}'} &= (3,6,8,9,9) \quad
&\mysum{}{\low{\mathbf{c}}} &= (2,4,4,4,4) \quad
&\mysum{}{\up{\mathbf{c}}'} &= (5,9,10). \quad
\end{alignat*}
The vector~$\mathbf{c}$ is initialized with~$\mathbf{c} \gets \low{\mathbf{c}} = (2,2,0,0,0)$. 
As~$\low{\mathbf{c}}' = (2,2,0)$ and~$\mysum{}{\low{\mathbf{c}}'} = (2,4,4)$, the number of steps of the outer loops is calculated by~$\delta_1 \gets \max\{ 4-2, 5-4, 5-4 \} = 2$.
\begin{enumerate}
	\item The inner loop breaks with~$i=5$. The left-most component equal to~$c_5=0$ is at position~$j=3$. Thus, the algorithm swaps~$\up{c}_5$ and~$\up{c}_3$ and increments~$c_3$. We gain~$\mathbf{c}=(2,2,1,0,0)$ and~$\up{\mathbf{c}}=(2,3,2,2,1)$.
	\item The inner loop breaks again with~$i=5$. The left-most component equal to~$c_5 = 0$ is now at position~$j=4$. Thus, the algorithm swaps~$\up{c}_5$ and~$\up{c}_4$ and increments~$c_4$. We gain~$\mathbf{c}=(2,2,1,1,0)$ and ~$\up{\mathbf{c}}=(2,3,2,1,2)$.
\end{enumerate}
The situation at the end of the first phase:
\begin{alignat*}{4}
\low{\mathbf{r}}  &= (4,1,0) \quad
&\up{\mathbf{r}}  &= (4,2,3) \quad
&\mathbf{c} &= (2,2,1,1,0) \quad
&\mathbf{c}' &= (4,2,0) \\
\mysum{}{\low{\mathbf{r}}} &= (4,5,5) \quad
&\mysum{}{\up{\mathbf{r}}} &= (4,6,9) \quad
&\mysum{}{\mathbf{c}} &= (2,4,5,6,6) \quad
&\mysum{}{\mathbf{c}'} &= (4,6,6).
\end{alignat*}
We observe that~$(\low{\mathbf{r}}, \up{\mathbf{r}}, \mathbf{c}, \mathbf{c})$ is realizable via Theorem~\ref{Theo:Schocker}.

\paragraph{Proof of Correctness}

The correctness of Alg.~\ref{Alg:PhaseOne} follows from Lemma~\ref{Theo:Incrementation} and~\ref{Theo:Termination}.  For the following theorems,
let~$\low{\mathbf{r}}=(\low{r}_1, \low{r}_2, \ldots, \low{r}_{n_1})$ and~$\mathbf{c} = (c_1, c_2, \ldots, c_{n_2})$ be non-increasing integer vectors and let~$\up{\mathbf{r}}=(\up{r}_1, \up{r}_2, \ldots, \up{r}_{n_1})$ and~$\up{\mathbf{c}} = (\up{c}_1, \up{c}_2, \ldots, \up{c}_{n_2})$ be integer vectors such that~$\low{\mathbf{r}} \leq \up{\mathbf{r}}$ and~$\mathbf{c} \leq \up{\mathbf{c}}$ hold,~$(\low{\mathbf{r}}, \up{\mathbf{r}}, \mathbf{c}, \up{\mathbf{c}})$ is realizable, and~$(\low{\mathbf{r}}, \up{\mathbf{r}}, \mathbf{c}, \mathbf{c})$ is not realizable.

\begin{lemma}
	There is a right-most position~$i$ with~$c_i < \up{c}_i$.
	Let~$j$ be the left-most position with~$c_j = c_i$.
	Let~$\mathbf{a} = (a_1, a_2, \ldots, a_{n_2})$ and~$\up{\mathbf{a}} = (\up{a}_1, \up{a}_2, \ldots, \up{a}_{n_2})$ be integer vectors defined by
	\begin{equation*}
	\begin{aligned}[c]
	a_k = \begin{cases}
	c_k + 1 ,& \textnormal{if } k = j,\\
	c_k ,& \textnormal{otherwise},
	\end{cases}
	\end{aligned}
	\qquad\textnormal{and}\qquad
	\begin{aligned}[c]
	\up{a}_k = \begin{cases}
	\up{c}_i ,& \textnormal{if } k = j,\\
	\up{c}_j ,& \textnormal{if } k = i,\\
	\up{c}_k ,& \textnormal{otherwise}.
	\end{cases}
	\end{aligned}
	\end{equation*}
	Then,~$\mathbf{a}$ is non-increasing,~$\mathbf{a} \leq \up{\mathbf{a}}$ holds, and~$(\low{\mathbf{r}}, \up{\mathbf{r}}, \mathbf{a}, \up{\mathbf{a}})$ is realizable.
	\label{Theo:Incrementation}
\end{lemma}

\begin{proof}
	
	As~$(\low{\mathbf{r}}, \up{\mathbf{r}}, \mathbf{c}, \up{\mathbf{c}})$ is realizable by assumption and~$(\low{\mathbf{r}}, \up{\mathbf{r}}, \mathbf{c}, \mathbf{c})$ is not,~$\mathbf{c}$ cannot be equal to~$\up{\mathbf{c}}$. Thus, there must be a position~$i$ with~$c_i < \up{c}_i$.
	As~$j$ is chosen left-most with~$c_j = c_i$, either~$j=1$ or~$c_{j-1} > c_{j}$ must hold. In both cases, the integer vector~$\mathbf{a}$ is non-increasing as~$\mathbf{c}$ is non-increasing by assumption. 
	
	We show next that~$\mathbf{a} \leq \up{\mathbf{a}}$ holds.
	As~$c_i < \up{c}_i$ holds by assumption, it follows that~$a_j = c_j + 1 = c_i + 1 \leq \up{c}_i = \up{a}_j$.
	To see that~$a_i \leq \up{a}_i$ holds, consider two cases.
	If~$i=j$, the inequality~$\up{a}_i = \up{c}_i \geq c_i + 1 = a_i$ holds as~$c_i < \up{c}_i$. Otherwise, if~$i \not= j$, we derive~$\up{a}_i = \up{c}_j \geq c_j = c_i = a_i$. Hence,~$a_i \leq \up{a}_i$ holds in both cases. Since~$a_k = c_k$ and~$\up{a}_k = \up{c}_k$ for all~$k \not= i,j$, we infer~$\mathbf{a} \leq \up{\mathbf{a}}$.
	
	It remains to show that~$(\low{\mathbf{r}}, \up{\mathbf{r}}, \mathbf{a}, \up{\mathbf{a}})$ is realizable. By Theorem~\ref{Theo:Schocker},~$(\low{\mathbf{r}}, \up{\mathbf{r}}, \mathbf{a}, \up{\mathbf{a}})$ is realizable if and only if~$\low{\mathbf{r}} \trianglelefteq \up{\mathbf{a}}'$ and~$\mathbf{a} \trianglelefteq \up{\mathbf{r}}'$ hold. Since~$\up{\mathbf{a}}$ is a permutation of~$\up{\mathbf{c}}$, the sets~$\{ \ell \colon \up{a}_i \geq k \}$ and~$\{ \ell \colon \up{c}_i \geq k \}$ are identical for all~$k$,
	thus,~$\up{\mathbf{a}}'$ equals~$\up{\mathbf{c}}'$. 
	As~$\low{\mathbf{r}} \trianglelefteq \up{\mathbf{c}}'$ holds due to the realizability of~$(\low{\mathbf{r}}, \up{\mathbf{r}}, \mathbf{c}, \up{\mathbf{c}})$, the dominance relation~$\low{\mathbf{r}} \trianglelefteq \up{\mathbf{a}}'$ holds, too.
	
	Now assume that~$\mathbf{a} \trianglelefteq \up{\mathbf{r}}'$ does not hold. Thus, there is a right-most position~$k$ such that~$\mysum{k}{\mathbf{a}} > \mysum{k}{\up{\mathbf{r}}'}$. 
	Since~$\mysum{k}{\mathbf{c}} \leq \mysum{k}{\up{\mathbf{r}}'}$ must hold as~$(\low{\mathbf{r}}, \up{\mathbf{r}}, \mathbf{c}, \up{\mathbf{c}})$  is realizable and~$\mathbf{a}$ differs from~$\mathbf{c}$ only at position~$j$, we conclude~$\mysum{k}{\mathbf{c}} = \mysum{k}{\up{\mathbf{r}}'}$ and~$j \leq k$.
	Next, we show that~$i > k$. Therefore, assume the contrary and let~$i \leq k$. Since~$i$ is chosen right-most, it follows that~$c_\ell = \up{c}_\ell$ for each~$\ell$ in range~$i < \ell \leq n_2$. Thus, increasing an arbitrary~$c_\ell$ in range~$i < \ell \leq n_2$ by a positive amount would violate~$\mathbf{c} \leq \up{\mathbf{c}}$ whereas increasing an arbitrary~$c_\ell$ in range~$1 \leq \ell \leq i$ violates the realizability of~$(\low{\mathbf{r}}, \up{\mathbf{r}}, \mathbf{c}, \up{\mathbf{c}})$. Hence,~$(\low{\mathbf{r}}, \up{\mathbf{r}}, \mathbf{c}, \up{\mathbf{c}})$ can only be realizable if~$(\low{\mathbf{r}}, \up{\mathbf{r}}, \mathbf{c}, \mathbf{c})$ is, which contradicts our assumption. Thus,~$i>k$.
	
	Since~$j$ was chosen left-most with~$c_j = c_i$, we obtain from~$j \leq k$ and~$k < i$ that~$c_j = \ldots = c_k = c_{k+1} = \ldots = c_i$. Since~$k$ is right-most and~$(\low{\mathbf{r}}, \up{\mathbf{r}}, \mathbf{c}, \up{\mathbf{c}})$ is realizable, we conclude that~$\mysum{k+1}{\mathbf{c}} < \mysum{k+1}{\up{\mathbf{r}}'}$ holds and thus,~$\up{r}_{k+1}' > c_{k+1}$.
	As~$\up{\mathbf{r}}'$ is non-increasing by definition, we conclude further that~$\up{r}_k' \geq \up{r}_{k+1}' > c_{k+1} = c_k$.
	Consequently,~$\mysum{k}{\mathbf{c}} = \mysum{k}{\up{\mathbf{r}}'}$ holds if and only if~$\mysum{k-1}{\mathbf{c}} > \mysum{k-1}{\up{\mathbf{r}}'}$ holds, which  contradicts the realizability of~$(\low{\mathbf{r}}, \up{\mathbf{r}}, \mathbf{c}, \up{\mathbf{c}})$ if~$k>1$, or is plainly wrong if~$k=1$.
	As a consequence,~$\mathbf{a} \trianglelefteq \up{\mathbf{r}}'$ must hold and thus,~$(\low{\mathbf{r}}, \up{\mathbf{r}}, \mathbf{a}, \up{\mathbf{a}})$ is realizable.
\end{proof}

\begin{lemma}
	Let~$\mathbf{c}^{(k)}$ and~$\up{\mathbf{c}}^{(k)}$ be the state of the integer vectors~$\mathbf{c}$ and~$\up{\mathbf{c}}$ after exactly~$k$ iterations of the outer loop in Alg.~\ref*{Alg:PhaseOne}.
	Then,~$(\low{\mathbf{r}}, \up{\mathbf{r}}, \mathbf{c}^{(k)}, \mathbf{c}^{(k)})$ is not realizable unless~$k \geq \delta_1 = \max\{ \mysum{j}{\low{\mathbf{r}}} - \mysum{j}{\low{\mathbf{c}}'} \colon 1 \leq j \leq n_1 \}$.
	\label{Theo:Termination}
\end{lemma}

\begin{proof}
	By Lemma~\ref{Theo:Incrementation},~$\mathbf{c}^{(k)}$ is non-increasing and~$\mathbf{c}^{(k)} \leq \up{\mathbf{c}}^{(k)}$ holds.
	Hence, by Theorem~\ref{Theo:Schocker}, the four-tuple~$(\low{\mathbf{r}}, \up{\mathbf{r}}, \mathbf{c}^{(k)}, \mathbf{c}^{(k)})$ is realizable if and only if~$\low{\mathbf{r}} \trianglelefteq (\mathbf{c}^{(k)})'$ and~$\mathbf{c}^{(k)} \leq \up{\mathbf{r}}'$ hold.
	Whereas the latter domination relation is ensured by Lemma~\ref{Theo:Incrementation}, the condition~$\low{\mathbf{r}} \trianglelefteq (\mathbf{c}^{(k)})'$ will not hold if~$k <\delta_1$.
	To see why this is true, consider incrementing an arbitrary component~$c^{(k)}_j$ to~$c^{(k+1)}_j = c^{(k)}_j + 1$.
	Whereas the components~$(c^{(k+1)})'_{\ell}$ will be equal to~$(c^{(k)})'_{\ell}$ for each~$\ell \not= c^{(k+1)}_j$, the component~$(c^{(k+1)})'_{\ell}$ will be of value~$(c^{(k+1)})'_{\ell} = (c^{(k)})'_{\ell} + 1$ for~$\ell = c^{(k+1)}_j$.
	As a consequence, the partial sums~$\mysum{}{\mathbf{c}^{(k+1)}}$ are affected by
	\[
	\mysum{i}{(\mathbf{c}^{(k+1)})'} = 
	\begin{cases}
	\mysum{i}{(\mathbf{c}^{(k)})'} + 1 ,&  \text{if } c^{(k)}_j < i \leq n_1,\\
	\mysum{i}{(\mathbf{c}^{(k)})'} ,& \text{otherwise}.
	\end{cases}
	\]
	Now let~$p$ be an arbitrary position such that ~$\mysum{p}{\low{\mathbf{r}}} > \mysum{p}{\low{\mathbf{c}}'}$ holds at the beginning of the first phase.
	Since~$\mysum{p}{\low{\mathbf{r}}}$ stays constant,
	the inequality~$\mysum{p}{\low{\mathbf{r}}} \leq \mysum{p}{(\mathbf{c}^{(k)})'}$ will be established as soon as a number of~$k \geq \mysum{p}{\low{\mathbf{r}}} - \mysum{p}{\mathbf{c}'}$ components of value~$c_j < p$ have been incremented.
	As Alg.~\ref{Alg:PhaseOne} chooses~$c_j$ as small as possible, the domination relation~$\low{\mathbf{r}} \trianglelefteq (\mathbf{c}^{(k)})'$ will hold after exactly~$\delta_1 = \max\{ \mysum{j}{\low{\mathbf{r}}} - \mysum{j}{\low{\mathbf{c}}'} \colon 1 \leq j \leq n_1 \}$ iterations.
\end{proof}

\paragraph{Running Time}
Determining the quantity~$\delta_1$ requires a running time of~$\mathcal{O}(n_1)$.
The outer loop of Alg.~\ref{Alg:PhaseOne} runs exactly~$\delta_1$ steps. In each step, the algorithm has to determine the position~$j \gets \min \{ \ell \colon c_{\ell} = c_i\}$ of the left-most occurrence of a component equal to~$c_i$. This can be achieved in constant time if we use a pre-computed lookup-table which is updated after each incrementation. Fortunately, each increment operations only requires a table-update at the positions~$c_i$ and~$c_i+1$ which can be executed in constant time. Thus, letting beside the inner loop, each step of the outer loop is executed in constant time. In contrast, the inner loop may need linear time. However, the variable~$i$ can be decreased at most~$n_2$ times during the whole process. Thus, the running time of the outer loop is~$\mathcal{O}(n_2 + \delta_1)$.
In summary, since~$\delta_1 \leq \mysum{n_1}{\low{\mathbf{r}}}$, the running time of the first phase can be bounded by~$\mathcal{O}(n_1 + n_2 + \mysum{n_1}{\low{\mathbf{r}}})$.

\subsection{Phase Two}

After the first phase has stopped,~$\low{\mathbf{r}}$ and~$\mathbf{c}$ are non-increasing integer vectors and~$(\low{\mathbf{r}}, \up{\mathbf{r}}, \mathbf{c}, \mathbf{c})$ is realizable by Lemma~\ref{Theo:Incrementation} and~\ref{Theo:Termination}.
Hence, there is a bipartite graph~$G=(U,V,E)$ with~$d_G(v_j) = c_j$ for~$j$ in range~$1\leq j \leq n_2$ and~$\low{r}_i \leq d_G(u_i) \leq \up{r}_i$ for~$i$ in range~$1 \leq i \leq n_1$.
If we switch the roles of vertex sets~$U$ and~$V$, we gain an instance of the realization problem in which the four-tuple~$(\mathbf{c}, \mathbf{c}, \low{\mathbf{r}}, \up{\mathbf{r}})$ is realizable,~$\mathbf{c}$ and~$\low{\mathbf{r}}$ are non-increasing, and~$\low{\mathbf{r}} \leq \up{\mathbf{r}}$ holds.
Thus, we can re-apply the first phase to the modified instance to construct a suitable integer vector~$\mathbf{r}$ such that~$(\mathbf{c}, \mathbf{c}, \mathbf{r}, \mathbf{r})$ is realizable. After switching back the roles of~$U$ and~$V$, we gain a bi-graphical pair~$(\mathbf{r}, \mathbf{c})$ of integer vectors. Afterwards, a realization is constructed by Ryser's algorithm.

\begin{algorithm}
	\LinesNumbered
	\DontPrintSemicolon
	\KwIn{realizable four-tuple~$(\low{\mathbf{r}}, \up{\mathbf{r}}, \mathbf{c}, \mathbf{c})$}
	\KwOut{bi-graphical pair~$(\mathbf{r},\mathbf{c})$}
	\caption{\textsc{Phase Two}\label{Alg:PhaseTwo}}
	$\mathbf{r} \gets \low{\mathbf{r}}$\tcp*{initialize~$\mathbf{r}$ with lower bounds~$\low{\mathbf{r}}$}
	$\delta_2 \gets \mysum{n_2}{\mathbf{c}} - \mysum{n_1}{\low{\mathbf{r}}}$\tcp*{calculate number of steps}
	$i \gets n_1$\tcp*{right-most position such that~$\low{r}_i < \up{r}_i$}
	\For{$k = 1, 2, \ldots, \delta_2$}{
		\While(\tcp*[f]{proceed to next position with~$r_i < \up{r}_i$}){$r_i = \up{r}_i$}{
			$i \gets i - 1$\;
		}
		$j \gets \min \{ \ell \colon r_{\ell} = r_i\}$ \tcp*{identify left-most~$r_j$ with~$r_j = r_i$}
		swap~$\up{r}_i$ and~$\up{r}_j$\;
		$r_j \gets r_j + 1$\;
	}
	\Return{$(\mathbf{r}, \mathbf{c})$}\;
\end{algorithm}	

\paragraph{Example}
We start where the first phase stopped.
\begin{alignat*}{4}
\low{\mathbf{r}}  &= (4,1,0) \quad
&\up{\mathbf{r}}  &= (4,2,3) \quad
&\mathbf{c} &= (2,2,1,1,0) \quad
&\mathbf{c}' &= (4,2,0) \\
\mysum{}{\low{\mathbf{r}}} &= (4,5,5) \quad
&\mysum{}{\up{\mathbf{r}}} &= (4,6,9) \quad
&\mysum{}{\mathbf{c}} &= (2,4,5,6,6) \quad
&\mysum{}{\mathbf{c}'} &= (4,6,6).
\end{alignat*}
The number of iterations is determined by~$\delta_2 \gets 6-5 = 1$.
\begin{enumerate}
	\item The inner loop breaks with~$i=3$. The left-most component equal to~$r_3 = 0$ is at position~$j=3$. Thus, the algorithm switches~$\up{r}_3$ and~$\up{r}_3$ and increments~$r_3$. We gain~$\mathbf{r}=(4,1,1)$ and~$\up{\mathbf{r}}=(4,2,3)$.
\end{enumerate}
The situation at the end of the second phase:
\begin{alignat*}{3}
\mathbf{r} &= (4,1,1) \quad
&\mathbf{c} &= (2,2,1,1,0) \quad
&\mathbf{c}' &= (4,2,0)\\
\mysum{}{\mathbf{r}} &= (4,5,6) \quad
&\mysum{}{\mathbf{c}} &= (2,4,5,6,6) \quad
&\mysum{}{\mathbf{c}'} &= (4,6,6)
\end{alignat*}
We verify by Theorem~\ref{Theo:GaleRyser} that~$(\mathbf{r}, \mathbf{c})$ is bi-graphical.
Fig.~\ref{Fig:ExampleRealization} shows a realization.
\begin{figure}
	\centering
	\includegraphics[]{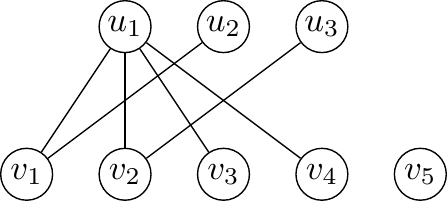}
	\caption{Realization of the sequence pair~$\mathbf{r} =  (4,1,1), \mathbf{c} = (2,2,1,1,0)$.}
	\label{Fig:ExampleRealization}
\end{figure} 

\paragraph{Proof of Correctness}
The correctness of the second phase can be shown very similarly to phase one and follows directly from the following theorem.

\begin{lemma}
	Let~$\mathbf{r}^{(k)}$ and~$\up{\mathbf{r}}^{(k)}$ be the state of the integer vectors~$\mathbf{r}$ and~$\up{\mathbf{r}}$ after exactly~$k$ iterations of the outer loop in Alg.~\ref*{Alg:PhaseTwo}.
	Then,~$(\mathbf{r}^{(k)}, \mathbf{r}^{(k)}, \mathbf{c}, \mathbf{c})$ is realizable if and only if~$k = \delta_2 = \mysum{n_2}{\mathbf{c}} - \mysum{n_1}{\low{\mathbf{r}}}$.
	\label{Theo:Termination2}
\end{lemma}

\begin{proof}
	By Theorem~\ref{Theo:Schocker}, the four-tuple~$(\mathbf{r}^{(k)}, \mathbf{r}^{(k)}, \mathbf{c}, \mathbf{c})$ is realizable if and only if~$\mathbf{r}^{(k)} \leq \mathbf{c}'$ and~$\mathbf{c} \trianglelefteq (\mathbf{r}^{(k)})'$ hold. In particular,~$\mysum{n_1}{\mathbf{r}^{(k)}} = \mysum{n_2}{\mathbf{c}}$ must hold.
	Since~$\mysum{n_1}{\low{\mathbf{r}}^{(k)}} \not= \mysum{n_2}{\mathbf{c}}$ for~$k \not= \mysum{n_2}{\mathbf{c}} - \mysum{n_1}{\low{\mathbf{r}}}$, the four-tuple~$(\low{\mathbf{r}}^{(k)}, \low{\mathbf{r}}^{(k)}, \mathbf{c}, \mathbf{c})$ cannot be realizable if~$k \not= \delta_2$.
	On the other hand, as~$(\mathbf{r}^{(k)}, \up{\mathbf{r}}^{(k)}, \mathbf{c}, \mathbf{c})$ is realizable and thus~$\mathbf{r}^{(k)} \trianglelefteq \mathbf{c}'$ holds after each iteration by Lemma~\ref{Theo:Incrementation}, the four-tuple~$(\mathbf{r}^{(k)}, \mathbf{r}^{(k)}, \mathbf{c}, \mathbf{c})$ will be realizable as soon as~$\mysum{n_1}{\mathbf{r}^{(k)}} = \mysum{n_2}{\mathbf{c}}$ holds, which is true after exactly~$\delta_2$ iterations.
\end{proof}

\paragraph{Running Time}

By similar arguments as before, the running time of the second phase can be described by~$\mathcal{O}(n_1 + n_2 + \mysum{n_2}{\mathbf{c}})$.
As Ryser's algorithm produces a realization~$G=(U,V,E)$ of the bi-graphical pair~$(\mathbf{r}, \mathbf{c})$ in~$\mathcal{O}(|U| + |V| + |E|)$ time and~$|E| = \mysum{n_2}{\mathbf{c}} \geq \mysum{n_1}{\low{\mathbf{r}}}$, the total running time of the realization algorithm is
$\mathcal{O}(|U| + |V| + |E|)$.
Since we cannot hope to construct a bipartite graph in sub-linear time, our algorithm is asymptotically optimal.

\subsection{Edge-Minimality}


\begin{lemma}
	The bipartite graph produced by our algorithm is edge-minimal.
\end{lemma}

\begin{proof}
	Let~$(\mathbf{r}, \mathbf{c})$ be a bi-graphical pair of integer vectors associated to an arbitrary realization~$G=(U,V,E)$ of the four-tuple~$(\low{\mathbf{r}}, \up{\mathbf{r}}, \low{\mathbf{c}}, \up{\mathbf{c}})$.
	As~$\low{\mathbf{c}} \leq \mathbf{c}$ must hold, the number of edges~$|E|$ is bounded from below by~$|E| \geq \mysum{n_2}{\low{\mathbf{c}}}$.
	In addition, if~$(\mathbf{r}, \mathbf{c})$ is bi-graphical, the four-tuple~$(\low{\mathbf{r}}, \up{\mathbf{r}}, \mathbf{c}, \mathbf{c})$ must be realizable.
	Hence, the inequalities~$\mysum{j}{\low{\mathbf{r}}} \leq \mysum{j}{\mathbf{c}'} = \mysum{j}{\low{\mathbf{c}}'} + x_j$ must hold for each~$j$ in range~$1 \leq j \leq n_1$. Thus,~$x_j = \mysum{j}{\low{\mathbf{r}}} - \mysum{j}{\low{\mathbf{c}}'}$ is the minimal number of edges that~$G$ needs to possess \emph{in addition to} the~$\mysum{n_2}{\low{\mathbf{c}}}$ edges, so that the inequality~$\mysum{j}{\low{\mathbf{r}}} \leq \mysum{j}{\mathbf{c}'}$ can hold.
	Hence, the total number of edges is bounded from below by~$|E| \geq \mysum{n_2}{\low{\mathbf{c}}} + \max \{ x_j \colon 1 \leq j \leq n_2 \} = \mysum{n_2}{\low{\mathbf{c}}} + \delta_1$.
	Since our algorithm produces a bipartite graph~$G=(U,V,E)$ with exactly~$|E| = \mysum{n_2}{\low{\mathbf{c}}} + \delta_1$ edges,~$G$ is edge-minimal.
\end{proof}
%
%

\paragraph{Remark}

Our algorithm can easily be used to construct \emph{edge-maximal} realizations.
For this purpose, consider an arbitrary edge-minimal realization~$G=(U,V,E)$ of the four-tuple~$(\low{\mathbf{r}}, \up{\mathbf{r}}, \low{\mathbf{c}}, \up{\mathbf{c}})$.
The associated \emph{complement graph}~$G^*=(U,V,E^*)$ is defined by~$E^* = (U \times V) \setminus E$. 
By construction, the graph~$G$ is edge-minimal if and only if~$G^*$ is edge-maximal.
In addition, it follows from definition that the degrees of $G^*$ are bounded from below and above by the \emph{complementary four-tuple}~$(\low{\mathbf{r}}^*, \up{\mathbf{r}}^*, \low{\mathbf{c}}^*, \up{\mathbf{c}}^*)$ with 
\begin{align*}
\low{r}^* &= (n_2 - \up{r}_1, n_2 - \up{r}_2, \ldots, n_2 - \up{r}_{n_1})\\
\up{r}^* &= (n_2 - \low{r}_1, n_2 - \low{r}_2, \ldots, n_2 - \low{r}_{n_1})\\
\low{c}^* &= (n_1 - \up{c}_1, n_1 - \up{c}_2, \ldots, n_1 - \up{c}_{n_2})\\
\up{c}^* &= (n_1 - \low{c}_1, n_1 - \low{c}_2, \ldots, n_1 - \low{c}_{n_2}).
\end{align*}
Thus, we can find an edge-maximal realization by first determining an edge-minimal realization of the complementary four-tuple~$(\low{\mathbf{r}}^*, \up{\mathbf{r}}^*, \low{\mathbf{c}}^*, \up{\mathbf{c}}^*)$ and creating the associated complement graph.

\section{Conclusion}

We gave a description of an algorithm that constructs a bipartite graph~$G=(U,V,E)$ whose degrees lie in prescribed intervals and showed that this algorithm has a running time of~$ \mathcal{O}(|U| + |V| + |E|)$. Since this is asymptotically optimal, the algorithm can be used to efficiently solve the realization problem.

\section*{Acknowledgements}

We thank Matthias M\"uller-Hannemann and Annabell Berger for their valuable suggestions and detailed remarks.
This research did not receive any specific grant from funding agencies in the public, commercial, or not-for-profit sectors.

\bibliographystyle{plain}

\begin{thebibliography}{1}
	
	\bibitem{anstee1985algorithmic}
	R.~P. Anstee.
	\newblock {An algorithmic proof of Tutte's f-factor theorem}.
	\newblock {\em Journal of Algorithms}, 6(1):112--131, 1985.
	
	\bibitem{fulkerson1959network}
	D.~R. Fulkerson.
	\newblock {A network flow feasibility theorem and combinatorial applications}.
	\newblock {\em Can. J. Math}, 11:440--451, 1959.
	
	\bibitem{gale1957theorem}
	D.~Gale.
	\newblock A theorem on flows in networks.
	\newblock {\em Pacific J. Math.}, 7(2):1073--1082, 1957.
	
	\bibitem{orlin2013max}
	James~B Orlin.
	\newblock {Max flows in O (nm) time, or better}.
	\newblock In {\em Proceedings of the forty-fifth annual ACM symposium on Theory
		of computing}, pages 765--774. ACM, 2013.
	
	\bibitem{rechner2017interval}
	S.~Rechner, L.~Strowick, and M.~Müller-Hannemann.
	\newblock {Uniform Sampling of Bipartite Graphs with Degrees in Prescribed
		Intervals}.
	\newblock {\em submitted to Journal of Complex Networks}.
	
	\bibitem{ryser1957combinatorial}
	Herbert~J Ryser.
	\newblock {Combinatorial properties of matrices of zeros and ones}.
	\newblock {\em Canadian Journal of Mathematics}, 9:371--377, 1957.
	
	\bibitem{schocker2001graphs}
	M.~Schocker.
	\newblock {On graphs with degrees in prescribed intervals}.
	\newblock In {\em Algebraic Combinatorics and Applications}, pages 307--315.
	Springer, 2001.
	
\end{thebibliography}

\end{document}